\documentclass[a4paper,conference]{IEEEtran}
\usepackage{graphicx,ifthen}
\usepackage{psfrag,amssymb,amsthm}
\usepackage[cmex10]{amsmath}
\usepackage{cite,flushend,color,pst-plot,stfloats,pst-sigsys}
\usepackage{pstricks-add}

\newtheorem{thm}{Theorem}
\newtheorem{lem}{Lemma}
\newtheorem{cor}{Corollary}

\newtheorem{prop}{Proposition}

\theoremstyle{definition}
\newtheorem{definition}{Definition}

\def \arxiv {1}

\title{Relative Information Loss -- An Introduction}

\author{\IEEEauthorblockN{Bernhard C. Geiger\IEEEauthorrefmark{1}, Gernot Kubin\IEEEauthorrefmark{1}
\IEEEauthorblockA{\IEEEauthorrefmark{1}Signal Processing and Speech Communication Laboratory, Graz University of Technology, Austria}
$\{$geiger,gernot.kubin$\}$@tugraz.at}}

\begin{document}
\newcounter{myTempCnt}

\ifthenelse{\arxiv=1}{
\newcommand{\x}[1]{x[#1]}
\newcommand{\y}[1]{y[#1]}

\newcommand{\pdfy}{f_Y(y)}

\newcommand{\ent}[1]{H(#1)}
\newcommand{\diffent}[1]{h(#1)}
\newcommand{\derate}[1]{\bar{h}\left(\mathbf{#1}\right)}
\newcommand{\mutinf}[1]{I(#1)}
\newcommand{\ginf}[1]{I_G(#1)}
\newcommand{\kld}[2]{D(#1||#2)}
\newcommand{\kldrate}[2]{\bar{D}(\mathbf{#1}||\mathbf{#2})}
\newcommand{\binent}[1]{H_2(#1)}
\newcommand{\binentneg}[1]{H_2^{-1}\left(#1\right)}
\newcommand{\entrate}[1]{\bar{H}(\mathbf{#1})}
\newcommand{\mutrate}[1]{\mutinf{\mathbf{#1}}}
\newcommand{\redrate}[1]{\bar{R}(\mathbf{#1})}
\newcommand{\pinrate}[1]{\vec{I}(\mathbf{#1})}
\newcommand{\loss}[2][\empty]{\ifthenelse{\equal{#1}{\empty}}{L(#2)}{L_{#1}(#2)}}
\newcommand{\lossrate}[2][\empty]{\ifthenelse{\equal{#1}{\empty}}{L(\mathbf{#2})}{L_{\mathbf{#1}}(\mathbf{#2})}}
\newcommand{\gain}[1]{G(#1)}
\newcommand{\atten}[1]{A(#1)}
\newcommand{\relLoss}[2][\empty]{\ifthenelse{\equal{#1}{\empty}}{l(#2)}{l_{#1}(#2)}}
\newcommand{\relLossrate}[1]{l(\mathbf{#1})}
\newcommand{\relTrans}[1]{t(#1)}
\newcommand{\partEnt}[2]{H^{#1}(#2)}

\newcommand{\dom}[1]{\mathcal{#1}}
\newcommand{\indset}[1]{\mathbb{I}\left({#1}\right)}

\newcommand{\unif}[2]{\mathcal{U}\left(#1,#2\right)}
\newcommand{\chis}[1]{\chi^2\left(#1\right)}
\newcommand{\chir}[1]{\chi\left(#1\right)}
\newcommand{\normdist}[2]{\mathcal{N}\left(#1,#2\right)}
\newcommand{\Prob}[1]{\mathrm{Pr}(#1)}
\newcommand{\Mar}[1]{\mathrm{Mar}(#1)}
\newcommand{\Qfunc}[1]{Q\left(#1\right)}

\newcommand{\expec}[1]{\mathrm{E}\left\{#1\right\}}
\newcommand{\expecwrt}[2]{\mathrm{E}_{#1}\left\{#2\right\}}
\newcommand{\var}[1]{\mathrm{Var}\left\{#1\right\}}
\renewcommand{\det}{\mathrm{det}}
\newcommand{\cov}[1]{\mathrm{Cov}\left\{#1\right\}}
\newcommand{\sgn}[1]{\mathrm{sgn}\left(#1\right)}
\newcommand{\sinc}[1]{\mathrm{sinc}\left(#1\right)}
\newcommand{\e}[1]{\mathrm{e}^{#1}}
\newcommand{\multint}{\iint{\cdots}\int}
\newcommand{\modd}[3]{((#1))_{#2}^{#3}}
\newcommand{\quant}[1]{Q\left(#1\right)}
\newcommand{\card}[1]{\mathrm{card}(#1)}
\newcommand{\diam}[1]{\mathrm{diam}(#1)}

\newcommand{\ivec}{\mathbf{i}}
\newcommand{\hvec}{\mathbf{h}}
\newcommand{\gvec}{\mathbf{g}}
\newcommand{\avec}{\mathbf{a}}
\newcommand{\kvec}{\mathbf{k}}
\newcommand{\fvec}{\mathbf{f}}
\newcommand{\vvec}{\mathbf{v}}
\newcommand{\xvec}{\mathbf{x}}
\newcommand{\Xvec}{\mathbf{X}}
\newcommand{\Xhvec}{\hat{\mathbf{X}}}
\newcommand{\xhvec}{\hat{\mathbf{x}}}
\newcommand{\xtvec}{\tilde{\mathbf{x}}}
\newcommand{\Yvec}{\mathbf{Y}}
\newcommand{\yvec}{\mathbf{y}}
\newcommand{\Zvec}{\mathbf{Z}}
\newcommand{\Nvec}{\mathbf{N}}
\newcommand{\Pvec}{\mathbf{P}}
\newcommand{\muvec}{\boldsymbol{\mu}}
\newcommand{\wvec}{\mathbf{w}}
\newcommand{\Wvec}{\mathbf{W}}
\newcommand{\Hmat}{\mathbf{H}}
\newcommand{\Amat}{\mathbf{A}}
\newcommand{\Fmat}{\mathbf{F}}

\newcommand{\zerovec}{\mathbf{0}}
\newcommand{\eye}{\mathbf{I}}
\newcommand{\evec}{\mathbf{i}}

\newcommand{\zeroone}{\left[\begin{array}{c}\zerovec^T\\ \eye\end{array} \right]}
\newcommand{\zerooneT}{\left[\begin{array}{cc}\zerovec & \eye\end{array} \right]}
\newcommand{\zerooneM}{\left[\begin{array}{cc}\zerovec &\zerovec^T\\\zerovec& \eye\end{array} \right]}

\newcommand{\Cxx}{\mathbf{C}_{XX}}
\newcommand{\Cx}{\mathbf{C}_{\Xvec}}
\newcommand{\Chx}{\hat{\mathbf{C}}_{\Xvec}}
\newcommand{\Cy}{\mathbf{C}_{\Yvec}}
\newcommand{\Cz}{\mathbf{C}_{\Zvec}}
\newcommand{\Cn}{\mathbf{C}_{\mathbf{N}}}
\newcommand{\Cnt}{\mathbf{C}_{\tilde{\mathbf{N}}}}
\newcommand{\Cxh}{\mathbf{C}_{\hat{X}\hat{X}}}
\newcommand{\rxx}{\mathbf{r}_{XX}}
\newcommand{\Cxy}{\mathbf{C}_{XY}}
\newcommand{\Cyy}{\mathbf{C}_{YY}}
\newcommand{\Cnn}{\mathbf{C}_{NN}}
\newcommand{\Cyx}{\mathbf{C}_{YX}}
\newcommand{\Cygx}{\mathbf{C}_{Y|X}}
\newcommand{\Wmat}{\underline{\mathbf{W}}}

\newcommand{\Jac}[2]{\mathcal{J}_{#1}(#2)}

\newcommand{\NN}{{N{\times}N}}
\newcommand{\perr}{P_e}
\newcommand{\perh}{\hat{\perr}}
\newcommand{\pert}{\tilde{\perr}}

\newcommand{\vecind}[1]{#1_0^n}
\newcommand{\roots}[2]{{#1}_{#2}^{(i_{#2})}}
\newcommand{\rootx}[1]{x_{#1}^{(i)}}
\newcommand{\rootn}[2]{x_{#1}^{#2,(i)}}

\newcommand{\markkern}[1]{f_M(#1)}
\newcommand{\pole}{a_1}
\newcommand{\preim}[1]{g^{-1}[#1]}
\newcommand{\preimV}[1]{\mathbf{g}^{-1}[#1]}
\newcommand{\Xmax}{\bar{X}}
\newcommand{\Xmin}{\underbar{X}}
\newcommand{\xmax}{x_{\max}}
\newcommand{\xmin}{x_{\min}}
\newcommand{\limn}{\lim_{n\to\infty}}
\newcommand{\limX}{\lim_{\hat{\Xvec}\to\Xvec}}
\newcommand{\limx}{\lim_{\hat{X}\to X}}
\newcommand{\limXo}{\lim_{\hat{X}_1\to X_1}}
\newcommand{\sumin}{\sum_{i=1}^n}
\newcommand{\finv}{f_\mathrm{inv}}
\newcommand{\ejtheta}{\e{\jmath\theta}}
\newcommand{\khat}{\bar{k}}
\newcommand{\modeq}[1]{g(#1)}
\newcommand{\partit}[1]{\mathcal{P}_{#1}}
\newcommand{\psd}[1]{S_{#1}(\e{\jmath \theta})}
\newcommand{\borel}[1]{\mathfrak{B}(#1)}
\newcommand{\infodim}[1]{d(#1)}

\newcommand{\delay}[2]{\psblock(#1){#2}{\footnotesize$z^{-1}$}}
\newcommand{\Quant}[2]{\psblock(#1){#2}{\footnotesize$\quant{\cdot}$}}
\newcommand{\moddev}[2]{\psblock(#1){#2}{\footnotesize$\modeq{\cdot}$}}}{}

\maketitle

\begin{abstract}
We introduce a relative variant of information loss to characterize the behavior of deterministic input-output systems. We show that the relative loss is closely related to R\'{e}nyi's information dimension. We provide an upper bound for continuous input random variables and an exact result for a class of functions (comprising quantizers) with infinite absolute information loss. A connection between relative information loss and reconstruction error is investigated.
\end{abstract}

\section{Introduction}\label{sec:intro}
System theory provides a vast literature of mathematical descriptions of deterministic input-output systems. The gain of a linear system at a specific input frequency is specified by its transfer function, and for the distortion introduced by nonlinear components certain single-letter measures (e.g., signal-to-distortion ratio) have been defined. These and the measures introduced for the design of systems (e.g., the mean-squared error) give ample choice to the engineer to characterize a system at hand. However, most of the available descriptions are energy-centered or consider second-order statistics only. A big exception are descriptions of chaotic, autonomous dynamical systems~\cite{Jost_DynSys}.

Recently, however, we observe a trend to employ information-theoretic descriptions and cost functions, especially in machine learning and nonlinear adaptive systems~\cite{Principe_ITLearning}. We believe that system theory would also benefit from single-letter information-theoretic characterizations of deterministic input-output systems, and thus have introduced \emph{information loss} as a possible candidate in~\cite{Geiger_ILStatic_IZS}. In this work we complement the notion of absolute information loss with its relative version, in order to provide a meaningful measure in cases where the absolute information loss is infinite.

Relative information loss for static functions, or \emph{fractional} information loss, has already been introduced by Watanabe~\cite{Watanabe_InfoLoss} in the context of stationary stochastic processes on finite alphabets. It is also worth mentioning that a rather similar quantity has been used in~\cite{Quinlan_GainRatio}, denoted as \emph{information gain ratio}:
\begin{equation}
 \frac{\mutinf{C;A}}{\ent{A}}
\end{equation}
There, $A$ is an attribute with a finite set of values, $C$ is a class variable, and the value of $A$ for which this measure achieves its maximum is assumed to be the most appropriate root of a decision tree used for classification. 
In this work we will consider the quantity
\begin{equation}
 1-\frac{\mutinf{C;A}}{\ent{A}} = \frac{\ent{A|C}}{\ent{A}}
\end{equation}
and extend its definition to a larger class of random variables.

The paper is organized as follows: We define relative information loss in Section~\ref{sec:defLoss} and analyze its elementary properties in Section~\ref{sec:basicProp}. Section~\ref{sec:constFunct} is devoted to a class of deterministic systems for which the absolute loss was shown to be infinite. We present a bound for the probability of a reconstruction error in Section~\ref{sec:Fano} and conclude with a few examples in Section~\ref{sec:examples}.


\section{A Definition of Relative Information Loss}\label{sec:defLoss}
We start with recalling the definition given in~\cite{Geiger_ILStatic_IZS}, where the absolute information loss induced by transforming an $N$-dimensional random variable (RV) $\Xvec$ to another $N$-dimensional RV $\Yvec$ by a static function $\gvec{:}\ \dom{X}\to\dom{Y}$, $\dom{X},\dom{Y}\subseteq\mathbb{R}^N$ was given as
\begin{equation}
 \loss{\Xvec\to \Yvec} = \sup_{\partit{}} \left(\mutinf{\hat{\Xvec};\Xvec}-\mutinf{\hat{\Xvec};\Yvec}\right) = \ent{\Xvec|\Yvec}\label{eq:loss}
\end{equation}
where the supremum is over all partitions $\partit{}$ of $\dom{X}$, and where $\hat{\Xvec}$ is obtained by quantizing $\Xvec$ according to the partition $\partit{}$ (see Fig.~\ref{fig:sysmod}). 

It was shown in~\cite{Geiger_ILStatic_IZS}, that there exist functions which loose an infinite amount of information; in particular, if the probability measure $P_\Xvec$ is absolutely continuous w.r.t. the $N$-dimensional Lebesgue measure ($P_\Xvec\ll\mu^N$), quantizers, limiters, and mappings to subspaces of lower dimensionality suffer from infinite information loss. Since some of these functions also transfer an infinite amount of information (i.e., $\mutinf{\Xvec;\Yvec}=\infty$), information loss alone obviously does not suffice to fully characterize the function $\gvec$ in information-theoretic terms.

Thus, we complement this absolute quantity of information loss by a relative one, indicating the percentage of information lost in the function:

\begin{definition}\label{def:relloss}
 Let $\Xvec$ be an $N$-dimensional RV on the sample space $\dom{X}$, and let $\Yvec$ be obtained by transforming $\Xvec$ with a static function $\gvec$. We define the \emph{relative information loss} induced by this transform as
\begin{equation}
 \relLoss{\Xvec\to\Yvec} = \limn \frac{\ent{\hat{\Xvec}_n|\Yvec}}{\ent{\hat{\Xvec}_n}}
\end{equation}
where $\hat{\Xvec}_n=\frac{\lfloor n\Xvec \rfloor}{n}$ (elementwise). The quantity on the left is defined if the limit on the right-hand side exists.
\end{definition}
One can consider $\hat{\Xvec}_n$ as being obtained by a vector quantization of $\Xvec$ with quantization bins equal to $N$-dimensional hypercubes of side length $\frac{1}{n}$ (i.e., using a uniform partition $\partit{n}$). Note that the partition $\partit{2^{k+1}}$ is a refinement of $\partit{2^k}$ ($\partit{2^{k+1}}\prec\partit{2^k}$). 

\textbf{Remark}: First of all, the limit of a sequence of increasingly fine quantizations now takes the place of the supremum in~\eqref{eq:loss}. (In Definition~\ref{def:relloss}, the supremum would lead to $\relLoss{\Xvec\to\Yvec}=1$.) Alternatively, as it was shown in~\cite{Geiger_ILStatic_IZS}, the limit of this sequence can also be used in the Definition of absolute information loss, i.e.,
\begin{equation}
 \loss{\Xvec\to \Yvec} = \limn \ent{\hat{\Xvec}_n|\Yvec} = \ent{\Xvec|\Yvec}.
\end{equation}
Again, this only holds if the limit exists.

\begin{figure}[t]
 \centering  
\begin{pspicture}[showgrid=false](0,1)(8,3.5)
 	\pssignal(0,2){x}{$\hat{\Xvec}$}
	\pssignal(4,2){xt}{$\Xvec$}
 	\pssignal(2,1){n}{$\partit{}$}
	\psfblock[framesize=1 0.75](2,2){oplus}{$Q(\cdot)$}
	\psfblock[framesize=1.5 1](6,2){c}{$g(\cdot)$}
	\pssignal(8,2){y}{$\Yvec$}
  \ncline[style=Arrow]{n}{oplus}
	\ncline[style=Arrow]{oplus}{x}
 \nclist[style=Arrow]{ncline}[naput]{xt,c,y}
	\nclist[style=Arrow]{ncline}[naput]{xt,oplus,x}
	\psline[style=Dash](1,2.75)(4,2.75)
	\psline[style=Dash](1,2.25)(1,2.75)
	\psline[style=Dash](4,2.25)(4,2.75)
	\psline[style=Dash](0.75,3.25)(7.25,3.25)
	\psline[style=Dash](0.75,2.25)(0.75,3.25)
	\psline[style=Dash](7.25,2.25)(7.25,3.25)
 	\rput*(2.5,2.75){\scriptsize{$\mutinf{\hat{\Xvec};\Xvec}$}}
	\rput*(4,3.25){\scriptsize{$\mutinf{\hat{\Xvec};\Yvec}$}}
\end{pspicture}
\caption{Model for computing the information loss of a memoryless input-output system $g$. $Q$ is a quantizer with partition $\partit{}$.}
\label{fig:sysmod}
\end{figure}
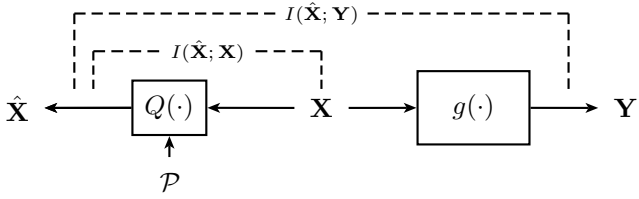

\section{Elementary Properties of Relative Information Loss}\label{sec:basicProp}
We will now highlight the basic properties of relative information loss: First of all, $\relLoss{\Xvec\to\Yvec}\in[0,1]$, which is due to the non-negativity of entropy and the fact that conditioning reduces entropy. It is interesting to note, however, that $\relLoss{\Xvec\to\Yvec}=0$ does not imply that the function $\gvec$ is \emph{information lossless}, i.e., that $\loss{\Xvec\to\Yvec}=0$. While this holds for discrete RVs $\Xvec$ (with finite entropy $\ent{\Xvec}$), for RVs with a continuous component this only means that the absolute information loss is finite. Conversely, $\relLoss{\Xvec\to\Yvec}=1$ does not imply that the information transfer $\mutinf{\Xvec;\Yvec}$ is zero. Again, while this holds for discrete RVs, for RVs with a continuous component $\relLoss{\Xvec\to\Yvec}=1$ implies a finite information transfer. However, we can state the following
\begin{prop}
 Let $\Xvec$ be such that $\ent{\Xvec}=\infty$ and let\\ $\relLoss{\Xvec\to\Yvec}>0$. Then, $\loss{\Xvec\to\Yvec}=\ent{\Xvec|\Yvec}=\infty$.
\end{prop}
\begin{IEEEproof}
We prove the proposition by contradiction. To this end, assume that $\ent{\Xvec|\Yvec}=L\leq\infty$. Then,
\begin{IEEEeqnarray}{RCL}
 \relLoss{\Xvec\to\Yvec} &=& \limn \frac{\ent{\hat{\Xvec}_n|\Yvec}}{\ent{\hat{\Xvec}_n}} =\limn\inf \frac{\ent{\hat{\Xvec}_n|\Yvec}}{\ent{\hat{\Xvec}_n}}\\
&\leq& \limn\inf\frac{\ent{\Xvec|\Yvec}}{\ent{\hat{\Xvec}_n}}\\ &=& \limn\inf\frac{L}{\ent{\hat{\Xvec}_n}}=0
\end{IEEEeqnarray}
where the inequality is due to data processing. The last equality follows from the fact that at least a subsequence of $\ent{\hat{\Xvec}_n}$ converges to $\ent{\Xvec}$ (cf.~\cite{Pinsker_InfoEngl,Kolmogorov_ContinuousRVs}).
\end{IEEEproof}

Another interesting property of the sequence 
\begin{equation}\label{eq:seq}
 \frac{\ent{\hat{\Xvec}_n|\Yvec}}{\ent{\hat{\Xvec}_n}}
\end{equation}
is that, while it might be converging (as we will show in some practically relevant cases below), it is neither generally increasing or decreasing. Consider, for example, a function $\gvec$ which is bijective if restricted to elements of the partition $\{\dom{X}_j\}$, but non-injective on its domain (cf.~\cite{Geiger_ILStatic_IZS}). Thus, for a partition $\partit{n_0}\prec\{\dom{X}_j\}$, and an input probability measure $P_\Xvec\ll\mu^N$ the sequence in~\eqref{eq:seq} is decreasing for all further refinements. Conversely, let $\gvec$ be a vector quantizer with partition $\{\dom{X}_j\}$ and let $\partit{n_0}=\{\dom{X}_j\}$. Here, while $\ent{\hat{\Xvec}_{n_0}|Y}=0$ the sequence in~\eqref{eq:seq} increases for all further refinements (cf.~Section~\ref{ssec:quantizers}).

Definition~\ref{def:relloss} has an interesting relationship to the $\epsilon$-entropy proposed by Kolmogorov in~\cite{Kolmogorov_ContinuousRVs,Kolmogorov_EpsilonEntropy}, but an even more tight connection can be made to the \emph{information dimension} proposed by R\'{e}nyi in~\cite{Renyi_InfoDim}. From there, we restate

\begin{lem}[Asymptotic behavior of $\ent{\hat{\Xvec}_n}$]\label{lem:epsilon}
Let $\Xvec$ be an RV with existing information dimension $\infodim{\Xvec}$ and let $\ent{\hat{\Xvec}_1}<\infty$. Then, for $n\to\infty$ the entropy of the RV $\hat{\Xvec}_n$ quantized as in Definition~\ref{def:relloss} behaves as
\begin{equation}
  \ent{\hat{\Xvec}_n} = \infodim{\Xvec}\log n + h + o(1)
\end{equation}
where $h$ is the $\infodim{\Xvec}$-dimensional entropy of $\Xvec$ (provided it exists).
\end{lem}
\begin{proof}
 See~\cite{Renyi_InfoDim} (cf. also~\cite{Kolmogorov_ContinuousRVs,Kolmogorov_EpsilonEntropy}).
\end{proof}
For an absolutely continuous RV $\Xvec$ we obtain from this Lemma the following
\begin{cor}[Theorems~1 \& 4 in~\cite{Renyi_InfoDim}]\label{cor:epsilon}
 Let $\Xvec$ be an RV with $P_\Xvec\ll\mu^N$ and $\ent{\hat{\Xvec}_1}<\infty$. Then, for $n\to\infty$ the entropy behaves as
\begin{equation}
 \ent{\hat{\Xvec}_n} = N\log n + \diffent{\Xvec} + o(1)
\end{equation}
where $\diffent{\cdot}$ is the differential entropy of $\Xvec$ (provided it exists).
\end{cor}

In other words, as a first approximation, the entropy of a continuous RV depends on the dimension of its probability measure, and only as a second approximation on the shape of its density. Note that the second and the third term in Lemma~\ref{lem:epsilon} can be neglected for large $n$.

Using these results we now maintain
\begin{thm}\label{thm:RILDim}
 Let $\Xvec$ be an RV with positive information dimension. Then, if $\infodim{\Xvec|\Yvec=\yvec}$ exists for all $\yvec\in\dom{Y}$, the relative information loss equals
\begin{equation}
 \relLoss{\Xvec\to\Yvec} = \frac{\expecwrt{\Yvec}{\infodim{\Xvec|\Yvec=\yvec}}}{\infodim{\Xvec}}
\end{equation}
where $\expecwrt{\Yvec}{\cdot}$ denotes the expectation w.r.t. $\Yvec$.
\end{thm}

\begin{IEEEproof}
 For the proof we use the definition of information dimension given in~\cite{Renyi_InfoDim},
\begin{equation}
 \infodim{\Xvec} = \limn \frac{\ent{\hat{\Xvec}_n}}{\log n}
\end{equation}
where by assumption the limit exists. We obtain
\begin{IEEEeqnarray}{RCL}
 \relLoss{\Xvec\to\Yvec} &=& \frac{\expecwrt{\Yvec}{\infodim{\Xvec|\Yvec=\yvec}}}{\infodim{\Xvec}}\\
&=&\frac{\int_\dom{Y} \limn \frac{\ent{\hat{\Xvec}_n|\Yvec=\yvec}}{\log n}dP_\Yvec(\yvec)}{\limn \frac{\ent{\hat{\Xvec}_n}}{\log n}}\\
&\stackrel{(a)}{=}&\frac{\limn \int_\dom{Y} \frac{\ent{\hat{\Xvec}_n|\Yvec=\yvec}}{\log n}dP_\Yvec(\yvec)}{\limn \frac{\ent{\hat{\Xvec}_n}}{\log n}}\\
&\stackrel{(b)}{=}&
\limn \frac{\int_\dom{Y}\ent{\hat{\Xvec}_n|\Yvec=\yvec}dP_\Yvec(\yvec)}{\ent{\hat{\Xvec}_n}}\\
&=&\limn \frac{\ent{\hat{\Xvec}_n|\Yvec}}{\ent{\hat{\Xvec}_n}}
\end{IEEEeqnarray}
where in $(a)$ we used Lebesgue's dominated convergence theorem (e.g.,~\cite{Rudin_Analysis3}) and where $(b)$ results from the fact that, by assumption, the limits in the numerator and denominator exist and are finite.
\end{IEEEproof}

This tight connection between relative information loss and the ratio of information dimensions leads to a series of interesting insights, as we will show in this and a companion paper~\cite{Geiger_RILPCA_arXiv}. In particular, it will prove useful if the probability measures are absolutely continuous w.r.t. Lebesgue measure, as information and geometric dimension coincide in this case (cf.~\cite{Renyi_InfoDim}).

We are now ready to establish an upper bound on the relative information loss in the following
\begin{thm}\label{thm:average}
 Let $\Xvec$ be an RV with a probability measure $P_\Xvec\ll\mu^N$ and with $\ent{\hat{\Xvec}_1}<\infty$. Then, if the quantities on the right exist,
\begin{equation}
 \relLoss{\Xvec\to\Yvec} \leq \frac{1}{N}\sum_{i=1}^N\relLoss{X^{(i)}\to\Yvec}\leq \frac{1}{N}\sum_{i=1}^N\relLoss{X^{(i)}\to Y^{(i)}}
\end{equation}
where $X^{(i)}$ and $Y^{(i)}$are the components of $\Xvec$ and $\Yvec$, respectively.
\end{thm}

\begin{proof}
 With Definition~\ref{def:relloss} and the chain rule of entropy we get
\begin{IEEEeqnarray}{RCL}
 \relLoss{\Xvec\to\Yvec}&=&\limn\frac{\sum_{i=1}^N\ent{\hat{X}^{(i)}_n|\hat{X}^{(1)}_n,\dots,\hat{X}^{(i-1)}_n,\Yvec}}{\ent{\hat{\Xvec}_n}}\notag\\
&\leq& \limn\frac{\sum_{i=1}^N\ent{\hat{X}^{(i)}_n|\Yvec}}{\ent{\hat{\Xvec}_n}}\label{eq:multline:lossbound}\\
&\stackrel{(a)}{=}& \frac{1}{\infodim{\Xvec}}\sum_{i=1}^N \expecwrt{\Yvec}{\infodim{X^{(i)}|\Yvec=\yvec}}
\end{IEEEeqnarray}
where in $(a)$ we exchanged summation and limit for similar reasons as in the proof of Theorem~\ref{thm:RILDim}. Corollary~\ref{cor:epsilon} now tells us that due to the absolute continuity $\infodim{\Xvec}=N$ and $\infodim{X^{(i)}}=1$ for all $i$. We thus obtain
\begin{multline}
 \relLoss{\Xvec\to\Yvec}\leq\frac{1}{N} \sum_{i=1}^N\frac{\expecwrt{\Yvec}{\infodim{X^{(i)}|\Yvec=\yvec}}}{\infodim{X^{(i)}}}\\=\frac{1}{N}\sum_{i=1}^N\relLoss{X^{(i)}\to\Yvec}
\end{multline}
which proves the first inequality. The second inequality is obtained by bounding $\ent{\hat{X}^{(i)}_n|\Yvec}\leq\ent{\hat{X}^{(i)}_n|Y^{(i)}}$ in~\eqref{eq:multline:lossbound}.
\end{proof}
At this point it is worth noting that throughout Section~\ref{sec:basicProp} no assumptions about a functional dependence between $\Xvec$ and $\Yvec$ were made. Indeed, all statements made in this Section hold equally for stochastic relationships (including stochastic independence) between $\Xvec$ and $\Yvec$.

\section{Relative Information Loss for Functions which are Constant}\label{sec:constFunct}
We now apply the relative information loss of Definition~\ref{def:relloss} to a class of functions for which we showed in~\cite{Geiger_ILStatic_IZS} that the absolute information loss is infinite. In particular, we are talking about functions which are constant on subsets $A_i$ of the domain with positive probability measure.

To this end, let $P_\Xvec\ll\mu^N$ be concentrated on a compact set $\dom{X}\subseteq\mathbb{R}^N$. Let further $A_i\subseteq\dom{X}$ with $P_\Xvec(A_i)>0$. Without loss of generality, we assume that the subsets $A_i$ are disjoint. Now take $\Yvec=\gvec(\Xvec)$, where $\gvec{:}\ \dom{X}\to\dom{Y}$, $\dom{Y}\subseteq\mathbb{R}^N$, is surjective, measurable, and constant on $A_i$, i.e., $\gvec(A_i)=\yvec_i$. As a consequence, $P_\Yvec$ is atomic on $\{\yvec_i\}$ (thus, $\loss{\Xvec\to\Yvec}=\infty$; cf.~\cite[Corollary~2]{Geiger_ILStatic_IZS}). With $A=\bigcup_i A_i$ we further require that $\gvec$ is piecewise bijective\footnote{see~\cite{Geiger_ILStatic_IZS} for a possible definition} on $\dom{X}\setminus A$, from which follows that $P_\Yvec$ is absolutely continuous on $\dom{Y}\setminus\{\yvec_i\}$.

We can now state the following
\begin{prop}\label{prop:functConst}
 Let $\Xvec$ be an RV with probability measure $P_\Xvec\ll\mu^N$ concentrated on a compact set $\dom{X}\subseteq\mathbb{R}^N$. Let $\gvec$ be such that it is constant on sets $A_i$ of positive $P_\Xvec$-measure and piecewise bijective elsewhere. Then, the relative information loss is
\begin{equation}
 \relLoss{\Xvec\to\Yvec} = P_\Xvec(A)
\end{equation}
where $A=\bigcup_i A_i$.
\end{prop}

\begin{proof}
By assumption and Corollary~\ref{cor:epsilon} we have $\infodim{\Xvec}=N$. Due to the properties of the function $\gvec$, $P_\Yvec$ decomposes into a component $P_\Yvec^{ac}\ll\mu^N$ and an atomic component $P_\Yvec^d$ concentrated on the points $\yvec_i=\gvec(A_i)$. Thus, the preimage of points $\yvec_i$ with positive $P_\Yvec$-measure is the union of the set $A_i$ and a countable number of points $\{x_{i,j}\}$. Since the set $A_i$ has positive $P_\Xvec$-measure (otherwise $P_\Yvec(\yvec_i)=0$), the conditional probability measure $P_{\Xvec|\Yvec=\yvec_i}\ll\mu^N$. Due to the compactness of the support $\dom{X}$ the conditional entropy $\ent{\hat{\Xvec}_1|\Yvec=\yvec_i}<\infty$, thus the associated information dimension exists and equals $N$. For all other points $\yvec\in\dom{Y}\setminus\{\yvec_i\}$ the preimage is a countable union of points. The associate conditional probability measure is 0-dimensional.

We now prove this Proposition with the help of Theorem~\ref{thm:RILDim}:
\begin{IEEEeqnarray}{RCL}
 \relLoss{\Xvec\to\Yvec} &=& \frac{1}{N}\int_\dom{Y} \infodim{\Xvec|\Yvec=\yvec} dP_\Yvec(\yvec)\\
&=& \frac{1}{N}\int_{\dom{Y}\setminus\{\yvec_i\}} \infodim{\Xvec|\Yvec=\yvec} dP_\Yvec(\yvec)\notag\\
&&+\frac{1}{N}\sum_i \infodim{\Xvec|\Yvec=\yvec_i} P_\Yvec(\yvec_i)\\
&=&\sum_i P_\Yvec(\yvec_i)
\end{IEEEeqnarray}
Since the preimage of $\yvec_i$ under $\gvec$ consists of a set $A_i$ of positive $P_\Xvec$-measure and (zero-measure) points, we can write
\begin{equation}
 \relLoss{\Xvec\to\Yvec} = \sum_i P_\Yvec(\yvec_i)= \sum_i P_\Xvec(A_i)= P_\Xvec(A)
\end{equation}
where the last equality follows from the fact that $A_i$ are disjoint and the additivity of the measure $P_\Xvec$.
\end{proof}

The interesting implication of this result is that the shape of the PDF on $A$ has no influence on the relative loss, and neither has the number of different sets $A_i$ (with different output values $\yvec_i$) -- yet, all these do have an influence on the information transport $\mutinf{\Xvec;\Yvec}$. This is in conflict with intuition, which suggests that whatever influences information transfer should also influence information loss, and, thus, also relative information loss. Yet, both the properties in Section~\ref{sec:basicProp} and the fact that $\ent{\hat{\Xvec}_n}$, as a first approximation, depends more on the dimension and the quantization bin size than on the shape of the PDF~\cite{Kolmogorov_ContinuousRVs} confirm this theoretical result.

Furthermore, in this particular case it turns out that $\Yvec$ is a mixture of a continuous and a discrete RV with information dimension $1-P_\Xvec(A)$~\cite{Renyi_InfoDim,Wu_Renyi}. One is thus led to the conjecture that indeed under some circumstances one can show that
\begin{equation}
 \relLoss{\Xvec\to\Yvec} = 1-\frac{\infodim{\Yvec}}{\infodim{\Xvec}}.
\end{equation}
If this really holds and under which conditions it does is currently under investigation.

\section{Relative Information Loss and Reconstruction Error}
\label{sec:Fano}
We next want to find connections between the relative information loss and the probability of a reconstruction error given by
\begin{equation}
 \perr= \min_f \Prob {\Xvec\neq f(\Yvec)}
\end{equation}
where $f$ is a function that tries to estimate or reconstruct the original $\Xvec$ from its image $\Yvec$.
It is well known that Fano's inequality does not hold for countably infinite alphabets (e.g.~\cite{Ho_EntrError}). However, we employ Fano's inequality here to derive a relationship between relative information loss and the probability of a reconstruction error by starting from a finite alphabet and then taking the limit. We present

\begin{thm}\label{thm:FanoRelative}
 Let $\Xvec$ be a RV with a probability measure $P_\Xvec\ll\mu^N$ which is concentrated on a compact set $\dom{X}\subset\mathbb{R}^N$. Let $\perr$ denote the probability of a reconstruction error. Then, the error probability is bounded by the relative information loss from below, i.e.,
\begin{equation}
 \perr\geq\relLoss{\Xvec\to\Yvec}.
\end{equation}
\end{thm}

\begin{IEEEproof}
 For the proof we start with a quantized version of the input RV, $\hat{\Xvec}_n$. Since $\hat{\Xvec}_n$ is a discrete RV on a finite alphabet $\hat{\dom{X}}_n$, we can employ the standard Fano bound~\cite{Cover_Information2},
\begin{equation}
 \ent{\hat{\Xvec}_n|\Yvec}\leq \binent{P_{e,n}}+P_{e,n}\log\card{\hat{\dom{X}}_n}
\end{equation}
where
\begin{equation}
 P_{e,n} = \Prob {\hat{\Xvec}_n\neq f^*(\Yvec)}.
\end{equation}
Since Fano's inequality holds for arbitrary estimators $f^*$, we let $f^*$ be the composition of $f^{\circ}=\arg\min_f\Prob {\Xvec\neq f(\Yvec)}$ and the quantizer of Definition~\ref{def:relloss}. $P_{e,n}$ is the probability that $f^{\circ}(\Yvec)$ and $\Xvec$ do not lie in the same quantization bin. Since the bin volume reduces with $n$, $P_{e,n}$ increases monotonically to $\perr$. We thus obtain with $\binent{p}\leq 1$ for all $0\leq p\leq 1$
\begin{equation}
 \ent{\hat{\Xvec}_n|\Yvec}\leq 1+\perr\log\card{\hat{\dom{X}}_n}.
\end{equation}

We next define the \emph{diameter} $D$ of $\dom{X}$ as
\begin{equation}
 D=\sup_{x_1,x_2\in\dom{X}}||x_1-x_2||
\end{equation}
where $||\cdot||$ is the Euclidean distance and where $D<\infty$ due to the compactness of $\dom{X}$. As an immediate consequence, $\dom{X}$ can be covered by an $N$-dimensional hypercube with side length $D$. Quantizing $\Xvec$ with a vector quantizer corresponds to covering $\dom{X}$ by hypercubes of side length $\frac{1}{n}$. It thus follows that
\begin{equation}
 \card{\hat{\dom{X}}_n} \leq\left(\lceil nD\rceil\right)^N\leq\left( nD+1\right)^N
\end{equation}
and finally
\begin{equation}
 \ent{\hat{\Xvec}_n|\Yvec}\leq 1+\perr N\log\left( nD+1\right).
\end{equation}
With Corollary~\ref{cor:epsilon} we thus get
\begin{IEEEeqnarray}{RCL}
 \relLoss{\Xvec\to\Yvec} &=& \limn \frac{\ent{\hat{\Xvec}_n|\Yvec}}{\ent{\hat{\Xvec}_n}}\\
&\leq& \limn \frac{1+\perr N\log\left( nD+1\right)}{\ent{\hat{\Xvec}_n}}\\
&\stackrel{(a)}{=}& \limn\frac{1}{N\log n}+\frac{\perr N\log\left( nD+1\right)}{N\log n}\\
&=&\perr
\end{IEEEeqnarray}
where in $(a)$ we again used Theorem~\ref{thm:RILDim} and the fact that $\infodim{\Xvec}=N$. This completes the proof.
\end{IEEEproof}

\section{Examples}\label{sec:examples}
In this Section we will now illustrate the theoretical results at the hand of a few examples.

\subsection{Quantizers}\label{ssec:quantizers}
The first practical application of our results will be the analysis of a quantizer, which is typically used to represent a continuous RV by a discrete RV, designed according to some optimality criterion (mean-squared reconstruction error, maximum output entropy, etc.). Since the output of the quantizer is discrete in amplitude, it is clear that an infinite amount of information is lost. In addition to that, since the quantizer function is constant almost everywhere it turns out that the relative information loss is unity:
\begin{equation}
 \relLoss{\Xvec\to\Yvec}=1
\end{equation}
In other words, disrespective of the (finite) number of quantization bins and the design criterion, the quantizer always destroys 100\% of the available information. This holds equally for scalar and vector quantizers. Note, however, that despite this fact still a positive amount of information is transferred by the quantizer (cf.~Section~\ref{sec:basicProp}).

\subsection{Center Clipper}\label{ssec:cclipper}
\begin{figure}[t!]
  \begin{center}
 \begin{pspicture}[showgrid=false](-2,-2)(2,2)
	\psaxeslabels{->}(0,0)(-2,-2)(2,2){$x$}{$g(x)$}
  \psplot[style=Graph,linecolor=black,plotpoints=500]{-1.8}{-0.85}{x}
	\psplot[style=Graph,linecolor=black,plotpoints=500]{0.85}{1.8}{x}
	\psplot[style=Graph,linecolor=black,plotpoints=500]{-0.8}{.8}{0}
	\psdisk[fillcolor=black](0.8,0){0.07}\psdisk[fillcolor=black](-0.8,0){0.07}
	\pscircle(0.8,0.8){0.07}\pscircle(-0.8,-0.8){0.07}
	\psTick{90}(0.8,0) \rput(0.8,-0.3){$c$}
	\psTick{90}(-0.8,0) \rput(-0.8,-0.3){$-c$}
 \end{pspicture}
\end{center}
\caption{Center clipper of Example 2}
\label{fig:cclipper}
\end{figure}
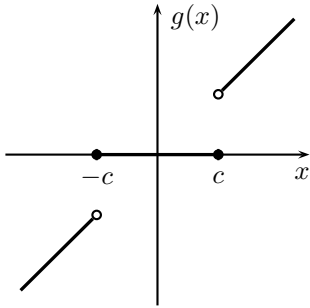
In signal processing center clippers (see Fig.~\ref{fig:cclipper}) are used for noise suppression or residual echo cancellation~\cite{Vary_DigSpeechTransm}. We let the center clipper be described by the following function:
\begin{equation}
 g(x) = \begin{cases}
 x,& \text{ if }|x|>c\\ 0, & \text{ else}
        \end{cases}
\end{equation}
By Theorem~\ref{prop:functConst} the relative information loss evaluates to $\relLoss{X\to Y}=P_X([-c,c])$, which reveals that it depends only on the clipping parameter $c$ and the probability mass contained in that interval. Yet, since center clippers \emph{do enhance} signal quality in many cases, this suggests that probably a different measure of information loss could be more appropriate.

Note further that the center clipper is bijective if it is restricted to $\dom{X}\setminus[-c,c]$. Thus, while outside of $[-c,c]$ we have a zero probability of a reconstruction error, within the center interval the error probability is unity. As a consequence, $\perr=P_X([-c,c])$ which makes the bound of Theorem~\ref{thm:FanoRelative} tight. If $g$ was not bijective outside of $[-c,c]$, but, e.g., would destroy the sign information, then $\perr>P_X([-c,c])$ and Theorem~\ref{thm:FanoRelative} still holds.

\section{Conclusion}
In this work, we introduced the notion of relative information loss, complementing its absolute variant presented by the authors in a previous work. We showed that there is a close connection between the relative loss and the R\'{e}nyi information dimension of the input and the conditional random variable of the input given the output.

For a continuous-valued input both upper bounds and an exact expression for a certain class of systems was presented. In particular, it was shown that quantizers loose 100\% of the available information. We finally analyzed a connection between the probability of reconstruction error and relative information loss.

\section*{Acknoledgment}
The authors thank Yihong Wu, Wharton School, University of Pennsylvania, for bringing R\'{e}nyi's information dimension to our attention.

\ifthenelse{\arxiv=1}{
\appendix
We now show the following
\begin{lem}\label{lem:quant}
\begin{equation}
 \limn \ent{\hat{\Xvec}_n|\Yvec} = \ent{\Xvec|\Yvec}
\end{equation}
provided the limit exists.
\end{lem}

\begin{IEEEproof}
 For the proof we note that
\begin{equation}
 \ent{\hat{\Xvec}_n|\Yvec} = \mutinf{\Xvec;\hat{\Xvec}_n|\Yvec}
\end{equation}
because $\hat{\Xvec}_n$ is a function of $\Xvec$~\cite[Ch.~3.9]{Pinsker_InfoEngl}. Further, if $\xi=(\xi_1,\xi_2,\dots)$ we obtain with~\cite[Thm.~3.10.1]{Pinsker_InfoEngl}
\begin{equation}
 \limn \mutinf{(\xi_1,\xi_2,\dots,\xi_n);\eta|\epsilon} = \mutinf{\xi;\eta|\epsilon}.\label{eq:pinsker}
\end{equation}
We now identify $\epsilon=\Yvec$ and $\eta=\Xvec$. Furthermore, if the limit in Lemma~\ref{lem:quant} exists, all subsequences converge to the same limit. In particular, also the subsequence $\hat{\Xvec}_{2^k}$ converges to the same limit. We now identify this RV with the binary expansion of $\Xvec$ up to order $k$; thus, $\hat{\Xvec}_{2^k}=(\xi_1,\xi_2,\dots,\xi_k)$. Clearly, $\lim_{k\to\infty}\hat{\Xvec}_{2^k} = \Xvec$. Comparing this to~\eqref{eq:pinsker} completes the proof.
\end{IEEEproof}}{}

\bibliographystyle{IEEEtran}
\bibliography{IEEEabrv,/afs/spsc.tugraz.at/project/IT4SP/1_d/Papers/InformationProcessing.bib,%
/afs/spsc.tugraz.at/project/IT4SP/1_d/Papers/ProbabilityPapers.bib,%
/afs/spsc.tugraz.at/user/bgeiger/includes/textbooks.bib,%
/afs/spsc.tugraz.at/user/bgeiger/includes/myOwn.bib,%
/afs/spsc.tugraz.at/user/bgeiger/includes/UWB.bib,%
/afs/spsc.tugraz.at/project/IT4SP/1_d/Papers/InformationWaves.bib,%
/afs/spsc.tugraz.at/project/IT4SP/1_d/Papers/ITBasics.bib,%
/afs/spsc.tugraz.at/project/IT4SP/1_d/Papers/HMMRate.bib,%
/afs/spsc.tugraz.at/project/IT4SP/1_d/Papers/ITAlgos.bib}

\begin{thebibliography}{10}
\providecommand{\url}[1]{#1}
\csname url@samestyle\endcsname
\providecommand{\newblock}{\relax}
\providecommand{\bibinfo}[2]{#2}
\providecommand{\BIBentrySTDinterwordspacing}{\spaceskip=0pt\relax}
\providecommand{\BIBentryALTinterwordstretchfactor}{4}
\providecommand{\BIBentryALTinterwordspacing}{\spaceskip=\fontdimen2\font plus
\BIBentryALTinterwordstretchfactor\fontdimen3\font minus
  \fontdimen4\font\relax}
\providecommand{\BIBforeignlanguage}[2]{{%
\expandafter\ifx\csname l@#1\endcsname\relax
\typeout{** WARNING: IEEEtran.bst: No hyphenation pattern has been}%
\typeout{** loaded for the language `#1'. Using the pattern for}%
\typeout{** the default language instead.}%
\else
\language=\csname l@#1\endcsname
\fi
#2}}
\providecommand{\BIBdecl}{\relax}
\BIBdecl

\bibitem{Jost_DynSys}
J.~Jost, \emph{Dynamical Systems: Examples of Complex Behavior}.\hskip 1em plus
  0.5em minus 0.4em\relax New York, NY: Springer, 2005.

\bibitem{Principe_ITLearning}
J.~C. Principe, \emph{Information Theoretic Learning: Renyi's Entropy and
  Kernel Perspectives}, ser. Information Science and Statistics.\hskip 1em plus
  0.5em minus 0.4em\relax New York, NY: Springer, 2010.

\bibitem{Geiger_ILStatic_IZS}
B.~C. Geiger and G.~Kubin, ``On the information loss in memoryless systems: The
  multivariate case,'' in \emph{Proc. Int. Zurich Seminar on Communications
  (IZS)}, Zurich, Feb. 2012, pp. 32--35, extended version available: {\tt
  arXiv:1109.4856 [cs.IT]}.

\bibitem{Watanabe_InfoLoss}
S.~Watanabe and C.~T. Abraham, ``Loss and recovery of information by coarse
  observation of stochastic chain,'' \emph{Information and Control}, vol.~3,
  no.~3, pp. 248--278, Sep. 1960.

\bibitem{Quinlan_GainRatio}
J.~R. Quinlan, ``Induction of decision trees,'' \emph{Machine Learning},
  vol.~1, pp. 81--106, 1986.

\bibitem{Pinsker_InfoEngl}
M.~S. Pinsker, \emph{Information and Information Stability of Random Variables
  and Processes}.\hskip 1em plus 0.5em minus 0.4em\relax San Francisco, CA:
  Holden Day, 1964.

\bibitem{Kolmogorov_ContinuousRVs}
A.~N. Kolmogorov, ``On the {Shannon} theory of information transmission in case
  of continuous signals,'' \emph{{IEEE} Trans. Inf. Theory}, vol.~2, pp.
  102--108, Dec. 1956.

\bibitem{Kolmogorov_EpsilonEntropy}
------, ``$\epsilon$-entropy and $\epsilon$-capacity of sets in functional
  spaces,'' in \emph{Selected Works of A. N. Kolmogorov -- Volume III:
  Information Theory and the Theory of Algorithms}, A.~N. Shiryayev, Ed.\hskip
  1em plus 0.5em minus 0.4em\relax Dordrecht: Kluwer, 1993, pp. 86--170.

\bibitem{Renyi_InfoDim}
A.~R\'{e}nyi, ``On the dimension and entropy of probability distributions,''
  \emph{Acta Mathematica Hungarica}, vol.~10, pp. 193--215, 1959.

\bibitem{Rudin_Analysis3}
W.~Rudin, \emph{Real and Complex Analysis}, 3rd~ed.\hskip 1em plus 0.5em minus
  0.4em\relax New York, NY: McGraw-Hill, 1987.

\bibitem{Geiger_RILPCA_arXiv}
B.~C. Geiger and G.~Kubin, ``Relative information loss in the {PCA},'' Apr.
  2012, {\tt arXiv:1202.???? [cs.IT]}.

\bibitem{Wu_Renyi}
Y.~Wu and S.~Verd\'{u}~and, ``R\'{e}nyi information dimension: Fundamental
  limits of almost lossless analog compression,'' \emph{{IEEE} Trans. Inf.
  Theory}, vol.~56, no.~8, pp. 3721--3748, Aug. 2010.

\bibitem{Ho_EntrError}
S.-W. Ho and S.~Verd\'{u}, ``On the interplay between conditional entropy and
  error probability,'' \emph{{IEEE} Trans. Inf. Theory}, vol.~56, no.~12, pp.
  5930--5942, Dec. 2010.

\bibitem{Cover_Information2}
T.~M. Cover and J.~A. Thomas, \emph{Elements of Information Theory},
  2nd~ed.\hskip 1em plus 0.5em minus 0.4em\relax Hoboken, NJ: Wiley
  Interscience, 2006.

\bibitem{Vary_DigSpeechTransm}
P.~Vary and R.~Martin, \emph{Digital speech transmission: {E}nhancement, coding
  and error concealment}.\hskip 1em plus 0.5em minus 0.4em\relax Chichester:
  John Wiley \& Sons, 2006.

\end{thebibliography}

\end{document}